\documentclass[sigconf]{acmart}

\usepackage{amsmath, amsthm}
\usepackage{listings, fancyvrb, multicol}
\usepackage{xcolor}
\usepackage{parcolumns}
\usepackage{tcolorbox}
\usepackage{balance}

\makeatletter
\lst@Key{countblanklines}{true}[t]%
{\lstKV@SetIf{#1}\lst@ifcountblanklines}

\lst@AddToHook{OnEmptyLine}{%
	\lst@ifnumberblanklines\else%
	\lst@ifcountblanklines\else%
	\advance\c@lstnumber-\@ne\relax%
	\fi%
	\fi}
\makeatother

\newtheorem{theorem}{Theorem}[section]
\newtheorem{observation}[theorem]{Observation}

\newcommand{\red}[1]{{\color{red}{#1}}}

\newcommand{\minsert}{multiInsert}
\newcommand{\mremove}{multiRemove}

\newcommand{\code}[1]{\texttt{#1}}
\newcommand{\oread}{\code{Read}}
\newcommand{\owrite}{\code{Write}}

\newcommand{\cas}{\code{CAS}}

\newcommand{\Naama}[1]{\noindent
	{\color{blue}{\textbf{Naama: }#1}}}

\newcommand{\descriptor}{descriptor}

\copyrightyear{2022}
\acmYear{2022}
\setcopyright{acmcopyright}\acmConference[PODC '22]{Proceedings of the 2022 ACM Symposium on Principles of Distributed Computing}{July 25--29, 2022}{Salerno, Italy}
\acmBooktitle{Proceedings of the 2022 ACM Symposium on Principles of Distributed Computing (PODC '22), July 25--29, 2022, Salerno, Italy}
\acmPrice{15.00}
\acmDOI{10.1145/3519270.3538448}
\acmISBN{978-1-4503-9262-4/22/07}



\begin{document}

\title{Fast and Fair Randomized Wait-Free Locks}

	\author{Naama Ben-David}
	\affiliation{
		\institution{VMware Research}
		\country{USA}}
	\email{bendavidn@vmware.com}

	\author{Guy Blelloch}
	\affiliation{
		\institution{Carnegie Mellon University}
		\country{USA}}
	\email{guyb@cs.cmu.edu}
	
	\begin{abstract}

We present a randomized approach for wait-free locks with strong bounds on time and fairness in a context in which any process can be arbitrarily delayed.  Our approach supports a tryLock operation that is given a set of locks, and code to run when all the locks are acquired.  A tryLock operation, or attempt, may fail if there is contention on the locks, in which case the code is not run.
Given an upper bound $\kappa$ known to the algorithm on the point contention of any lock, and an upper bound $L$ on the number of locks in a tryLock's set,  a tryLock will succeed in acquiring its locks and running the code with probability at least $1/(\kappa L)$.  It is thus fair. Furthermore, if the maximum step complexity for the code in any lock is $T$, the attempt will take $O(\kappa^2 L^2 T)$ steps, regardless of whether it succeeds or fails.  The attempts are independent, thus if the tryLock is repeatedly retried on failure, it will succeed in $O(\kappa^3 L^3 T)$ expected steps, and with high probability in not much more.



\end{abstract}
	\maketitle


		\section{Introduction}

In concurrent programs, \emph{locks} allow executing a `critical section' of code atomically, so that it appears to happen in isolation. Locks are likely the most important primitives in concurrent and distributed computing; they give the illusion of a sequential setting, thereby simplifying program design. 
%
However, locks can also become scalability bottlenecks for concurrent systems. 

To illustrate these concepts, we use the classic dining philosophers problem as a running example.  In the \emph{dining philosophers problem}, first introduced by Dijkstra, $n$ philosophers sit around a table, with one chopstick between each pair of philosophers.  Each philosopher is in one of three states -- thinking, hungry or eating -- when hungry, they need to pick up both adjacent chopsticks to be able to eat.  When done eating, they put down the chopstick and think for an unpredictable amount of time before next being hungry.  In the asynchronous setting, a scheduler decides when each philosopher takes a step, and can delay a philosopher arbitrarily.  What should philosophers do to minimize the number of steps they take from becoming hungry until 
they are done eating?  To avoid having the philosophers starve, we make two assumptions: firstly, once a philosopher acquires the chopsticks, the number of steps taken eating is bounded by a constant, and secondly, others can help a philosopher eat by taking steps on their behalf.  Without the first, a philosopher could starve their neighbor by eating forever, and without the second, the scheduler could starve a philosopher by never letting its chopstick-holding neighbor take a step.

In this paper we present the first algorithm for this setting that ensures that each philosopher will acquire their chopsticks and eat in $O(1)$ steps in expectation.    
Clearly, our algorithm thus ensures progress is made quickly (at least in the asymptotic sense), and ensures fairness in the sense that everyone who is hungry gets to eat.
Our algorithm is randomized, and assumes an oblivious adversarial scheduler (i.e., one that decides the interleaving of the philosopher's steps ahead of time), but adaptive adversarial philosophers who can choose how long to think knowing everything about the system.  

We are not just interested in the dining habits of philosophers, but more generally in fine-grained locks.  The chopsticks represent the locks, the philosophers processes, and eating represents a critical section of code.  By allowing arbitrary code in the critical section, more than two locks per critical section, and arbitrary conflicts among the locks (not just neighbors on a cycle), the setting covers a wide variety of applications of fine-grained locks.  For example, it captures operations on linked lists, trees, or graphs that require taking a lock on a node and its neighbors for the purpose of making a local update.  Indeed, such local updates with fine-grained locks are the basis of a large number of concurrent data structures~\cite{drachsler14,bronson10,masstree12,lazylist06,KungL80,Bayer88,arttre16,pugh89}, and of graph processing systems such as GraphLab~\cite{low2010graphlab}. Note that in many of these applications, the number of locks that need to be taken is still a small constant.

Our approach relies on light-weight \emph{tryLock} attempts, which may fail and then be retried.
Specifically, a tryLock specifies a set of locks to acquire, and code to run in the critical section.  If a tryLock attempt by a process successfully acquires the locks, the given code is run.  In this paper, critical sections can contain arbitrary code involving private steps along with reads, writes and CAS operations on shared memory.\footnote{While most locks do not allow critical sections to experience races, we allow such scenarios, for more general group-locking mechanisms.}  We do not allow nesting of locks---i.e., the critical code cannot contain another tryLock.  The critical section ends with a call to \emph{release}, which releases all the process's locks.  For mutual exclusion, we assume that if two processes have acquired the same lock, it must appear (based on updates to shared memory) that their critical sections did not overlap in time.\footnote{We say ``appear'' since helping can cause instructions to overlap in time, but those instructions will have no effect on the shared state.}

Our main contribution is the following result.

\begin{theorem}[Informal] \label{thm:main} 
	Let $\kappa$ be the maximum number of tryLock attempts on any lock at any given time, let $L$ be the maximum number of locks per tryLock attempt, and let $T$ be the maximum number of steps taken by a critical section. There exists an algorithm for wait-free fine-granularity locks against an oblivious scheduler with the following properties.
	\begin{itemize}
		\item \textbf{Step Bound.} Each tryLock attempt takes $O(\kappa^2 L^2 T)$ steps.
		\item \textbf{Fairness Bound.} Each tryLock attempt of a process $p$ succeeds with probability at least $\frac{1}{\kappa L}$ independently of $p$'s other attempts, and allowing for a adaptive adversary to decide when to attempt the tryLock.
	\end{itemize}
\end{theorem}

Since attempts by a process $p$ are independent, a direct corollary of this result is a wait-free fine-grained lock algorithm that succeeds  in expected $O(\kappa^3 L^3 T)$ steps; simply retry upon failure.
This is the first result that achieves any step complexity bounds that rely only on these parameters. 
As a special case, our results imply an $O(1)$ step solution to the dining philosophers problem; that is, assuming it takes constant steps to eat, each attempt to eat succeeds with probability 1/4 and takes $O(1)$ steps to complete (here, $\kappa = L = 2$). 

We first describe an algorithm that assumes knowledge of the bounds $\kappa$ and $L$, and then remove this assumption using a guess-and-double technique at the cost of a logarithmic loss in success probability.
That is, for a tryLock with bounded contention $\kappa$ that is \emph{unknown} to the algorithm, the probability of success is $\Omega\left(\frac{1}{\kappa L \log (\kappa L T)}\right)$. 

Achieving non-blocking progress requires processes to help each other complete their critical sections. 
This may result in several processes running the same code when helping concurrently. Thus, to ensure correctness, critical sections must be made \emph{idempotent}, that is, regardless of how many times they are run, they appear to have only executed once. The notion of idempotence in computer science has been recognized as useful in various contexts~\cite{idempotence02,idempotence12,idempotence13}. Turek, Shasha and Prakash~\cite{turek1992locking} and Barnes~\cite{barnes1993method}, designed lock-free locks and 
showed how to make any code based on non-concurrent reads and writes idempotent with constant overhead.
However, they did not distill this property, and instead present their protocols as ad-hoc ways to use lock-free locks.

In this paper, we formally define idempotence for concurrent programs, and present a new and more general construction to achieve it.  In particular, the Turek, Shasha and Prakash's, and Barnes' approaches only supported reads and writes in critical sections, and only if they did not race.  We allow for races and also support CAS.    As with theirs, our approach only has a constant factor overhead.
Thus, our wait-free locks are applicable to general code without any asymptotic overhead. We believe that the definition of idempotence and its new construction are of independent interest; indeed our formalization of idempotence has recently been used elsewhere~\cite{ben2022lock}.


       \section{Assumptions and Approach}\label{sec:overview}

For the purpose of outlining the key assumptions, challenges and approach, and
to generalize beyond dining philosophers while still abstracting away details of the machine model, we consider acquiring locks as a game involving players and competitions.  Each player (tryLock attempt or philosopher) $p$ participates in the game by specifying a set of competitions (locks or chopsticks) it wants to participate in.  Different players can specify different (but potentially overlapping) sets of competitions.  While a player is in the game, it takes steps to try to win its competitions, and a scheduler interleaves the steps of the different players.  If a player wins all its competitions (acquires all its locks), it wins the game and celebrates (executes its critical section or eats).  The celebration is itself a sequence of steps.  It then exits.  To ensure mutual exclusion, no two players can simultaneously have won their game (before exiting) if they share a competition.

\paragraph{Adversarial assumptions.}
An \emph{adversarial scheduler} is often used to model the inherent asynchrony in concurrent systems; that is, the order in which processes execute steps (the \emph{schedule}) is assumed to be controlled by an adversary. An \emph{adaptive adversary} is assumed to see everything that has happened in the execution thus far (the history), whereas an \emph{oblivious adversary} is assumed to make all of its scheduling decisions before the execution begins. Some separations are known between adaptive and oblivious scheduler settings~\cite{Abrahamson88,Rabin80,giakkoupis2012tight,giakkoupis2014randomized}. Often, an oblivious adversary is considered to be a reasonable assumption, since asynchrony in real hardware is not generally affected by the values written on memory. 

In our setting, in addition to the adversarial scheduler, the players can be adversarial, possibly trying to increase or decrease the probability of celebrating.  In the setting of tryLocks it is unreasonable to assume
that the players are oblivious having pre-decided when to enter the game and with what competitions.  This is because a player will likely need to try again if it fails on the first attempt, and possibly on a different set of locks.  Therefore for all but the very simplest protocols the point at which a player requests to enter a game, and possibly which competitions it requests to compete on, will depend on what has happened so far.  We therefore assume the player is adaptive and makes decisions knowing the full history.  In summary, we assume two separate adversaries; the \emph{player adversary}, which is \emph{adaptive} and controls the start time and set of locks of each tryLock attempt, and the \emph{scheduler adversary}, which is \emph{oblivious} and controls the order of player steps.

\paragraph{Random Priorities.}
In our algorithm, as in other algorithms for similar problems~\cite{Rabin80,LR81,LynchSS94,DFP02}, each player is assigned a \emph{random priority} such that higher priorities win over lower priorities.  In the synchronous setting this by itself ``solves'' our problem.  In particular in a round a set of players play the game by (1) picking their priority, (2) checking if they are the highest priority on all their competitions, and (3) celebrating if they are.  Assuming independent and unique priorities, the probability of a player $p$ celebrating is at least inversely proportional to the number of distinct players that requested any of the competitions $p$ played on.  The game is therefore fair, and assuming a synchronous scheduler (round-robin), the number of steps is bounded.

Unfortunately in the asynchronous setting, even with an oblivious scheduler, the situation is much more difficult.  Firstly, bounding the steps requires having players help others; otherwise, a player could be blocked waiting for another player to celebrate.  We solve this using idempotent code.  Secondly, and much more subtly, keeping the competition fair is challenging against an adaptive player adversary.  Most obviously, if the player adversary wants a player to lose it could wait for other strong players to be in shared competitions (recall that it can see the history), and then start the player.  Even if we hide the priorities from the adversary, it could likely gain knowledge by how the players are doing.  Even an oblivious player adversary could skew the game; if strong players take more steps than weak players and stay active longer, incoming players would see a biased field of strong players.  There are several other subtle difficulties with achieving fairness.

\paragraph{Our approach.}
Our approach to making the game fair is to ensure the adversary's choices introduce no bias once a player enters the game.
We prevent the introduction of biased priorities by the adversaries with two key ideas.
Firstly, each player enters the game in a \emph{pending} state, before its priority is assigned. Before being assigned a priority, a player $p$ must help complete the attempts of all the competitors that started before it. That is, any player $p'$ whose priority was known to the player adversary before player $p$ joined the game will be forced to finish competing without competing against $p$. After completing this helping phase, $p$ is assigned its priority, in what we call its \emph{reveal} step; after this step, $p$ is no longer pending, and is now \emph{active}. 

To implement the reveal step atomically, we model each competition (lock) as an \emph{active set} object, which keeps track of membership (i.e. who is currently competing on this competition/lock). It allows players to insert and remove themselves, as well as query the object to get the set of currently active members. We then model the system of competitions as a \emph{multi active set} object, which allows players to insert themselves into the membership of several competitions at once, i.e., all the competitions they compete on. 
The active state is then easy to implement; players are considered active if they appear in the active set of their locks. 

However, helping others before becoming active is not sufficient to mitigate the bias that the adversary could introduce; 
the player adversary can enter a new player $p'$ after $p$ but while $p$ is still pending, and the player $p'$ could overtake $p$ and become active before $p$ (recall the scheduler adversary can make $p'$ much faster than $p$).  Based on steps in the protocol, $p'$'s priority could affect the step at which $p$ reaches its reveal step. This allows the adversary to affect whether $p$ competes with $p'$ based on the latter's priority. 

To avoid this problem, our second idea is to force the player to stay in the pending state for a fixed number of its own steps, before revealing itself and competing. We do so by introducing \emph{delays} in which the player simply stalls until it has taken enough steps since it joined the game.
The important aspect of introducing these fixed delays is that the time a player becomes active (and thus begins to compete) is unaffected by other players.  It therefore cannot be sucked earlier or pushed later based on the priority of current competitors.

We note that our approach is robust against strong adversarial players, but only an oblivious scheduler.  A strong scheduler could still move the point at which a player reveals itself based on known priorities.  We leave handling an adaptive scheduler adversary as an open question.



       \section{Related Work}
\label{sec:related}

\paragraph{Randomization in Locks.}
Using randomization to acquire locks is a difficult problem that has been studied for many years.  The difficulty arises from the lack of synchronization among processes, and the ability of the adversary to delay processes based on observations of the current competition.  Rabin~\cite{Rabin80} first considered the problem for a single lock, and only for the acquisition (i.e. no helping).  Like ours, his scheme used priorities.  Saias~\cite{Saias92} showed the algorithm did not satisfy the claimed fairness bounds due to information leaks of the sort described in Section~\ref{sec:overview}, and Kushilevitz and Rabin~\cite{kushilevitz1992randomized} fixed it with a more involved algorithm.  Lehmann and Rabin 
also developed an algorithm for the dining philosophers problem~\cite{LR81}.  Lynch, Saias and Segala ~\cite{LynchSS94} later proved that with probability 1/16 within 9 rounds one philosopher would eat.   However, our goal is much stronger, requiring a constant fraction to eat.  Moreover, their model is not fully asynchronous---a round involves every process taking a step.  Duflot, Fribourg, and Picaronny generalized the algorithm to the fully asynchronous setting~\cite{DFP02}, but at the cost of a bound that depends on the number of processors.

More recent work has also looked at randomized mutual exclusion for a single lock and without helping~\cite{giakkoupis2012tight,giakkoupis2014randomized}.  This work has focused on the Distributed Shared Memory (DSM) model, which separates local from remote memory accesses.  It shows that although the local time is necessarily unbounded (since there is no helping), the number of remote accesses can be bounded.  Most of the above work has assumed an adversary that knows what instructions have been run on each process, but not the arguments of those instructions.  This is more powerful than an oblivious adversary, but less powerful than an adaptive one.  None of the work considered separating the player adversary from the scheduler adversary.

It has been shown that a tenet of concurrent algorithm design, linearizability~\cite{Herlihy90} does not nicely extend when randomization is introduced~\cite{GHW11}.  Linearizability allows operations that take multiple steps to be treated as if they run atomically in one step.  Unfortunately, however, analyzing probability distributions for the single step case does not generalize to a multistep linearized implementation, especially when analyzed for a weaker adversary.     This is what lead to some of the difficulties encountered by the previous work, and some of the challenges we face.

\paragraph{The Need for Randomization.}
It seems unlikely that acquiring wait-free fine-granularity locks can be done in $O(1)$ steps deterministically.  This is even true in the simpler case of the dining philosophers, where each philosopher only ever tries to acquire  two locks and each lock only ever has two philosophers contending on it. Assuming the philosophers do not know their position around the table, this is even true in a \emph{synchronous} setting; the best known solution to the similar two-ruling-set problem, i.e., identifying a subset of $n$ philosophers who are separated by one or two other philosophers, takes $O(\log^* n)$ steps~\cite{ColeV86}.  

The issue is that there is a symmetry that needs to be broken.  In the randomized synchronous setting, the problem becomes easy using random priorities as discussed in Section~\ref{sec:overview}---every philopher will have the highest
priority among itself and its two neighbors with probability $1/3$ and will therefore eat with that probability.
In the asynchronous setting, the problem can be solved in $O(n)$ steps deterministically using, for example, Herlihy's universal wait-free construction~\cite{Herlihy91}.  Every philosopher can announce when they are hungry and then try to help all others in a round robin manner, using a shared pointer to the philosopher currently being helped.
Using more sophisticated constructions~\cite{Afek97}, the steps can be reduced, but the total number of steps still depends on the total number of concurrently hungry philosophers in the system.  

\paragraph{Lock-Free Locks.}
Turek et al.~\cite{turek1992locking} and independently Barnes~\cite{barnes1993method} introduced the idea of lock-free locks.  They are both based on the idea of leaving a pointer to code to execute inside the locks, such that others can help complete it.  In the locked code, Turek et al.'s method supports reads and writes and locks nested inside each other.  As with standard locks, cycles in the inclusion graph must be avoided to prevent deadlocks. 
Their approach thus allows arbitrary static transactions via two-phase locking by ordering the locks, and acquiring them in that order.  It uses recursive (or ``altruistic'') helping in that it recursively helps transactions encountered on a required lock.  It is lock free, uses CAS, and although the authors do not give time bounds it appears that if all transactions take at most $T_m$ time in isolation, the amortized time per transaction is $O(p T_m)$, where $p$ is the number of processes.
Barnes's approach supports arbitrary dynamic transactions in a lock free manner, and uses LL/SC.  
It uses a form of optimistic concurrency~\cite{KungR81} allowing for dynamic transactions.
As with the Turek et al. approach, it uses recursive helping.  Neither approach is wait-free---a transaction can continuously help and then lose to yet another transaction. 

To allow our locks to be non-blocking, we present a general construction for achieving idempotence. A similar construction was recently presented and 
used it to implement lock-free locks~\cite{ben2022lock}. Their lock-free locks, while efficient in practice, do not have a bound on steps per tryLock attempt, as a single attempt can help arbitrarily many other ongoing attempts (not necessarily only on the locks in its lock set).
They are lock-free, but not wait-free.

\paragraph{Contention Management in Transactions.}
Shavit and Touitou~\cite{shavit1997software} introduce the idea of ``selfish'' helping in the context of transactions.  They argue that if a transaction encounters a lock that is taken, it should help the occupant release this lock, but not recursively help.  In particular, if while helping another transaction, it encounters a taken lock, then it aborts the transaction being helped. 
Their approach only supports static transactions, as it needs to take locks in a fixed order.  It differs from our helping scheme in that there are no priorities involved.  In our scheme, when a transaction being helped meets another transaction on a lock, we abort the one with lower priority and continue with the one being helped if it is not the one aborted.  Shavit and Touitou's approach is again lock-free but not wait-free.  The worst case time bounds are weaker than Turek et al. or Barnes since there can be a chain of aborted transactions as long as the size of memory, where only the last one succeeds.  

Fraser and Harris~\cite{fraser2007concurrent} extend Barnes's approach based on optimistic concurrency and recursive helping.  The primary difference is that they avoid locks for read-only locations by using a validate phase (as originally suggested by Kung and Robinson~\cite{KungR81}).  They break cycles between a validating read and a write lock on the same location, by giving arbitrary (not random) priorities to the transactions to break this cycle.  As with Shavit and Touitou's method, operations can take amortized time proportional to the size of memory.
There has been a variety of work on contention management for transactions under controlled schedulers, some of it using
randomization~\cite{AttiyaEST06,GuerraouiHP05,SchneiderW09,SharmaB12}, but it does not apply to the asynchronous setting we are considering.

\paragraph{Efficient Wait-Freedom.}
Starting with Herlihy~\cite{Herlihy91}, many researchers have studied wait-free universal constructions, many of which can be applied
to at least a single lock, but most of these have an $O(P)$ factor in their time complexity, where $P$ is the \emph{total number of processes in the system}, meaning that even under low contention they are very costly.
Afek et al.~\cite{afek1995wait} describe an elegant solution for a universal construction, or single lock in our terminology, that reduced the time complexity to be proportional to the point contention instead of the number of processors.   
Attiya and Dagan~\cite{AttiyaD01} describe a technique that should be able to support nested locks, although described in terms of operations on multiple locations.  They only support accessing two locations (i.e., two locks).  Considering the conflict graph among live transactions, they describe an algorithm such that when transactions are separated by at least $O(\log^* n)$ in the graph, they cannot affect each other.  The approach is lock free, but not wait free, and no time bounds are given. 
Afek et al.~\cite{Afek97} generalize the approach to a constant $k$ locks (locations) and describe a wait-free variant using a Universal construction.  They show that the step complexity (only counting memory operations) 
is bounded by a function of the contention within a neighborhood of radius $O(\log^* n)$ in the conflict graph.  
Both approaches are very complicated due to their use of a derandomization technique for breaking symmetries~\cite{ColeV86}.

       \section{Model and Preliminaries}\label{sec:model}

In this paper we use standard operations on memory including \oread{}, \owrite{}, and \code{CAS}.  
Beyond memory operations, processes do local operations (e.g. register operations, jumps, ...).  Whenever we discuss the execution \emph{time} for a process, we mean \textbf{all} operations (i.e., instructions) that run on the process including the local ones.

A \emph{procedure} is a sequencial procedure with an invocation point (possibly with arguments), and a response (possibly with return values).   A \emph{step} is either a memory operation or an invocation or response of a procedure.  We assume all steps are annotated with their arguments and return values, and we say two steps are \emph{equivalent} if these are the same.  
We say a memory operation has \emph{no effect} if it does not change the memory (e.g. a read, a failed CAS or a write of the same value).  We assume the standard definition of linearizability~\cite{Herlihy90}.

The \emph{history} of a procedure is the sequence of steps it took, or has taken so far.  The history of a concurrent program is some interleaving of the histories of the individual procedures.  A history is valid if it is consistent with the semantics of the memory operations.  

A \emph{thunk} is a procedure with no arguments~\cite{ingerman1961thunks}.  Note that any code (e.g., the critical section of a lock)
can be converted to a thunk by wrapping it along with its free variables into a 
closure~\cite{steele1976lambda} (e.g., using a lambda in most modern programming languages).  Here, for simplicity, we also assume thunks do not return any values---they can instead write a result into a specified location in memory. 
A thunk runs with some local private memory, and accesses the main memory via a fixed set of memory operations.  

A \emph{lock} object $\ell$ provides a \emph{$tryLock_\ell$} operation, which returns a boolean value; if false, we say the $tryLock_\ell$ fails, and if it returns true, then we say that the $tryLock_\ell$ succeeded. We also define a general \emph{tryLock} procedure whose arguments are a \emph{lock set}, i.e., a set of lock and a \emph{thunk}.  We call the execution of a trylock a \emph{tryLock attempt}. A tryLock calls $tryLock_\ell$ on each of the locks in its set, and succeeds if and only if all of them succeed. A tryLock attempt $p$ returns a boolean value indicating whether it succeeded or failed, and satisfies mutual exclusion as defined in Definition~\ref{def:mutex}.


In this paper, we aim to construct \emph{randomized wait-free} locks, and furthermore bound running time and success probabilities in terms of the \emph{point contention} on the locks in the system. We say an algorithm is \emph{randomized wait-free} if each process takes a finite expected number of steps until its operation succeeds (see Chor et al~\cite{chor1994wait} for a more formal definition). In this paper, the operations of processes are tryLock attempts. We say a tryLock attempt is \emph{live} on a lock $\ell$ from its invocation to its response (inclusive) if $\ell$ is in its lock set. The \emph{point contention of lock $\ell$ at time $t$} is the number of live tryLock attempts at time $t$ that contain $\ell$ in their lock set. The \emph{maximum point contention}, $k_\ell$ of lock $\ell$ is the maximum point contention $\ell$ can have at any point in time across all possible executions. We let $\kappa$ be an upper bound on $k_\ell$ for all locks $\ell$ in the system.
The contention of a tryLock attempt $p$, $C_p$, is the sum across all of $p$'s locks of $k_\ell$. That is, $C_p = \sum_{\ell \in p.lockList} k_\ell$.


We assume two adversaries; an \emph{adaptive player adversary} and an \emph{oblivious scheduler adversary}. Formally, the scheduler adversary is a function from a time step to the process that runs an instruction on that time step, which produces a \emph{schedule}.  Our algorithms do not know this function, and the scheduler can delay any given process for an arbitrary length of time. The player adversary is a function from the history of an execution and a given process to a boolean indicating whether the given process starts a new tryLock at its next step, and if so with which locks.


\subsection{Idempotence}



To allow processes to help each other complete their thunks (critical sections) on a lock, we must ensure that regardless of how many processes execute a thunk, it only appears to execute once. For this, we use the notion of \emph{idempotence}, which roughly means that a piece of code that is applied multiple times appears as if was run once~\cite{idempotence02,idempotence13,idempotence12,ben2019delay}. We define the notion of idempotence here, and show how to make any thunk involving \oread, \owrite{} and compare-and-swap (\cas{}) instructions into one that is idempotent with constant overhead in the full version of the paper~\cite{lltheory}. 
Barnes~\cite{barnes1993method} and Turek et al.~\cite{turek1992locking} do not extract the notion of idempotence, but do describe a way to make code based on non-concurrent reads and writes idempotent under our definition (below).


A \emph{run} of a thunk $T$ is the sequence of steps taken by \emph{a single process} to execute or help execute $T$.  The runs for a thunk can be interleaved.
A run is \emph{finished} if it reached the end of $T$.
We say a sequence of steps $S$
is \emph{consistent} with a run \emph{r} of \emph{T} if, ignoring process ids, $S$ contains the exact same steps as $r$.
We use $E~|~T$ to denote the result of starting from an execution $E$ and removing any step that does not belong to a run of the thunk $T$.


\begin{definition}[Idempotence]
	\label{def:idempotence}
		A thunk $T$ is idempotent if in any valid execution $E$ consisting of runs of $T$ interleaved with arbitrary other steps on shared memory, there exists a subsequence $E'$ of $E | T$ such that:
	\begin{enumerate}
		\item
		if there is a finished run of $T$, then the last step of the first such finished run must be the end of $E'$,
		\item
		removing all of $T$'s steps from $E$ other than those in $E'$ leaves a valid history consistent with a single run of $T$.
	\end{enumerate}
%
%
%
\end{definition}
The definition essentially states that the combination of all runs of a thunk $T$ is equivalent to having run
$T$ once, and finishing at the response of the first run. 

In the full version of the paper~\cite{lltheory} we describe a simulation/translation that converts thunks involving \oread, \owrite{} and \cas{} 
instructions into one that is idempotent, proving the following result.


\begin{theorem}
\label{theorem:idempotent}
Any thunk using only \oread{}, \owrite{} and \cas{} operations on shared memory can be simulated using \oread{}s, \owrite{}s and \cas{es} as primitive operations such that (1) it is idempotent,  (2) every simulated memory operation takes constant time, (3) the simulated operations are linearizable. 
\end{theorem}
We note that recent concurrent work proved a similar result~\cite{ben2022lock}, but our simulation technique is different and we believe it is of independent interest.

\subsection{Mutual Exclusion with Idempotence}

We say the \emph{interval} of an idempotent thunk is from the first step of any run of the thunk until the last step of the first run that completes.
When implementing a tryLock, the safety property we require is as follows. 

\begin{definition}[Mutual Exclusion with Idempotence] \label{def:mutex}
	If a tryLock attempt $\mathcal{A}$ with thunk $\mathcal{T}$ and lock set $\mathcal{L}$ succeeds, then there is a run of $\mathcal{T}$  
that executed to completion. Furthermore, $\mathcal{T}'s$ interval does not overlap the interval of any other thunk 
whose lock set overlaps~$\mathcal{L}$.
If $\mathcal{A}$ fails, there is no run of $\mathcal{T}$.
\end{definition}


       \section{A Multi Active Set Algorithm}\label{sec:activeset}


We now define the \emph{multi active set} problem and present an algorithm that solves it using an \emph{active set} object. 
Multi active sets will be useful in implementing our locking scheme; in Section~\ref{sec:nestedLock}, we show how to implement fast and fair locks by representing them as a multi active set object.

The \emph{active set} object was first introduced by Afek et al.~\cite{afek1999long}. It has three operations; \emph{insert}, \emph{remove}, and \emph{getSet}. A getSet operation returns the set of elements that have been inserted but not yet removed. Insert and remove operations simply return `ack', and processes must alternate insert and remove calls. 

The \emph{multi active set} problem is a generalization of the active set problem to multiple sets.  Instead of an insert, the data structure supports a \emph{\minsert} that inserts an \emph{item} into a collection of sets.  The \emph{\mremove} operation is with respect to the previous \minsert{} operation, and removes the item from the sets it was inserted into.   The getSet operation takes a set as an argument, and behaves the same as for the active set, returning all items in the given set.

\subsection{Active Set Algorithm}
We now present a novel linearizable active set algorithm. Its pseudocode appears in Algorithm~\ref{alg:activeSet}. In the full version of the paper~\cite{lltheory}, we prove it correct and show that its step complexity is \emph{adaptive} to the size of the set; insert and remove operations take $O(k)$ steps for a set with $k$ elements, and the getSet operation takes constant time.

An \code{announcements} array of $C$ slots is maintained, where $C$ is the maximum number of elements that can be in the set at any given time. Each slot has an \code{owner} element and a \code{set}, which is a pointer to a linked-list of elements. To insert an element, a process traverses the announcements array from the beginning, looking for a slot whose \code{owner} field is empty. It then takes ownership of this slot by CASing in its new element into the slot's owner field. To remove an element from slot $i$, a process simply changes the owner of slot $i$ to \code{null}. We assume that a process maintains the index of the slot that it successfully owned in its last insert, and uses this index in its next call to remove. Intuitively, the owner fields of all the slots make up the current active set. An insert operation can always find a slot without an owner, since there are $C$ slots.

To help implement an efficient linearizable getSet function, the insert and remove operations propagate the changed ownership of their slot to the top of the announcements array by calling the \code{climb} helper function.
The \code{climb} function works as follows. Starting at the slot given as an argument, it traverses the announcements array to the top, replacing the \code{set} field of the current slot $i$ with the \code{set} of the previous slot $i+1$, plus the owner of slot $i$. That is, the \code{climb} function intuitively collects all owners of the slots and propagates all of them to the \code{set} field of slot $0$. The getSet function can then simply read the \code{set} of \code{announcements[0]} to get the current active~set.

This algorithm is similar to the universal construction presented by Afek et al. in STOC'95~\cite{afek1995wait}, and is \emph{adaptive};
the step complexity of the insert and remove operations is proportional to the size of the active set plus the point contention during the insert operation in a given execution. This is because the number of slots that an insert operation traverses before finding one with no owner is at most the number of elements currently in the active set, plus the ones in the process of being inserted. 
We note that when using the active set object to count membership in a larger context (as was its original intent and is the way we use it for the lock algorithm), this translates to the point contention in the larger context.

\renewcommand{\figurename}{Algorithm}
\begin{figure}[h]
	\begin{lstlisting}
	struct Slot:
	T owner
	T* set
	Slot[C] announcements
	
	climb(int i):
	for j = i ... 0:
	for k = 1 ... 2:
	curSet = announcements[j].set @\label{line:readCurSet}@
	if j == C: newSet = announcements[j].set //corner case
	else:	newSet = announcements[j+1].set
	newMember = announcements[j].owner
	if newMember != null:
	newSet += newMember
	CAS(announcements[j].set, curSet, newSet) @\label{line:climbCAS}@
	
	T* getSet():
	return announcements[0].set @\label{line:returnSet}@
	
	int insert(T p):
	for i = 0... C-1:
	if CAS(announcements[i].owner, null, p):
	climb(i)
	return i
	
	remove(int i):
	announcements[i].owner = null
	climb(i)
	\end{lstlisting}
	\caption{Active Set Algorithm}
	\label{alg:activeSet}
\end{figure}

\subsection{Making a Multi Active Set}
We now present an implementation of a multi active set that relies on an active set object. However, our multi active set object is not linearizable. Instead, we require a weaker property, which will suffice for our use of the multi active set object to implement locks. In particular, 
%
every \minsert{} and \mremove{}  must appear to happen atomically at some point between the invocation and response.  Any getSet operation that is invoked after that point, will see the effect of the operation, and any that responds before that point will not see the effect of the operation.  However, any getSet that overlaps the point might or might-not see the effect.    This property is remeniscent of \emph{regularity} as defined for registers by Lamport~\cite{lamport1986interprocess}; we therefore call it \emph{set regularity}. 

We show how to implement a set-regular multi active set from a linearizable active set object in Algorithm~\ref{alg:multiactive}. Each item is endowed with a flag that is initialized to false.  To \emph{\minsert} an item into a given collection of sets, the item is first inserted into each of these sets using an active-set insert, and then its flag is set to true.  The \mremove{} operation first unsets the flag, and then removes the item from each of the sets.  The getSet operation for the multi active set first calls the active set getSet operation, and then scans the items, returning
the only ones for which it sees the flag is set to true.   The flags can be scanned in any order, which implies the getset operation
is not linearizable.    For example, items $a$ and $b$ could be inserted into a set by two separate \minsert{}s, and 
for two getset operations that overlap the insert, one could return just $a$ and the other just $b$.

\renewcommand{\figurename}{Algorithm}
\begin{figure}
\begin{lstlisting}
type T:
	void setFlag()
	void clearFlag()
	bool getFlag()

void multiInsert(T item, ActiveSet* collection):
	item.clearFlag()
	for set in collection: set.insert(item)
	item.setFlag()
	
void multiRemove(T item, ActiveSet* collection):
	item.clearFlag()
	for set in collection: set.remove(item)

T* getSet(ActiveSet A):
	T* set =  A.getSet()
	for T in set:
		if not T.getFlag():
			remove T from set
	return set
\end{lstlisting}
\caption{Multi Active Set Algorithm}
\label{alg:multiactive}
\end{figure}

We prove the following correctness and step complexity theorems in the full version of the paper~\cite{lltheory}.

\begin{theorem}\label{thm:multisetCorrectness}
	The Multi Active Set algorithm presented in Algorithm~\ref{alg:multiactive} satisfies set regularity, assuming it uses a linearizable Active Set implementation.
\end{theorem}

\begin{theorem}\label{thm:multisetTime}
	In Algorithm~\ref{alg:multiactive}, each operation takes $O(\kappa)$ steps per active set it accesses.
\end{theorem}

\section{The Lock Algorithm}\label{sec:nestedLock}

We now present the lock algorithm, whose pseudocode is shown in Algorithm~\ref{alg:lock}. Intuitively, each lock is represented by an active set object that is part of a single multi active set object. Each tryLock attempt creates a \emph{\descriptor}, which specifies the list of locks to be acquired, the code to run if the locks are acquired successfully, and two other metadata fields: the \emph{priority} assigned to this descriptor, and its current \emph{status}. The status is set to \code{active} initially, and can be changed to \code{lost} or \code{won} later in the execution as the fate of this attempt is determined. The descriptor is used as the item to be inserted into the active sets; the priority field doubles as the flag for the multi active set; initially, it is set to $-1$, indicating a false flag. 

After initializing its \descriptor{}, a process starts its tryLock \emph{attempt} with that \descriptor. In a slight abuse of notation, we sometimes use a \descriptor{} $p$ to refer to the process that initialized $p$ as it is executing the attempt tied to $p$.  Without loss of generality we assume that each attempt is tied to a unique \descriptor.  
At  a high level, the algorithm implements each lock as an active set object, using Algorithm~\ref{alg:activeSet} (described in Section~\ref{sec:activeset}). A  \descriptor{} is inserted into the active sets of each of its locks via a  \minsert{}; to set the descriptor's flag to true in the multi active set algorithm, the negative priority is replaced with a uniformly randomly chosen value.
The \descriptor{} then calls \code{getSet} on each of its locks in turn, and compares its priority to that of all other \descriptor{s} in the set. Intuitively, if a \descriptor{} $p$ has the maximum priority of all \descriptor{s} on all of its locks, then it wins, and its thunk gets executed (celebrates).

However, the algorithm is more subtle, as it must block the adversary from skewing the distribution of a given \descriptor{} $p$'s competitors.
Therefore, upon starting a new tryLock attempt, before calling the  \minsert{}, 
and in particular, before choosing a random priority, 
$p$ \emph{helps} all \descriptor{s} on its locks.  
Intuitively, this is done to `clear the playing field' by ensuring that any \descriptor{} whose priority might have affected $p$'s adversarial start time cannot compete with $p$.
To help other \descriptor{s}, $p$ executes a \code{getSet} on each of its locks in turn, and, for each \descriptor{} $p'$ in the set, 
$p$ helps $p'$ determine whether it will win or lose. To do so, it calls the \code{run(p')} function, which serves as the helping function and is the way that a \descriptor{} competes against other \descriptor{s}. We describe the \code{run} function in more detail below; this function is the core of the lock algorithm. Before describing it, we first explain what a \descriptor{} $p$ does after helping, and when it calls the \code{run} function to help itself.

After having executed the \code{run} function for every competitor, 
it is time for $p$ to enter the game itself. First, $p$ calls the  \minsert{} with its lock set as the argument. Recall that before returning, the  \minsert{} sets $p$'s priority to a uniformly random value.\footnote{In this work, we assume that priorities do not conflict. To enforce this, it suffices to pick priorities in a range that is polynomial in the total number of processes, $P$, in the system. Conflicts can be handled by considering both processes to have lost, and would only slightly affect our bounds.}
 We call the time at which $p$'s priority is written its \emph{reveal step}, since it now reveals its priority to all other \descriptor{s}, and can now start receiving help from others.
 Note that by  the set regularity property of the multi active set and the priority's use as $p$'s flag in the multi active set, 
any \code{getSet} on one of those locks that starts its execution after $p$'s reveal step will return a set that includes $p$.
$p$ now calls \code{run(p)} to compete in the game. 

After returning from the \code{run(p)} call, $p$ is guaranteed to have a non-active status (either \code{won} or \code{lost}). That is, it knows the outcome of its attempt.  At  this point, $p$ cleans up after itself by calling multiRemove to remove itself from all active sets it was in. 

\paragraph{The \code{run} function.} The \code{run} function forms the core of the lock algorithm.
The \code{run} function on a \descriptor{} $p$ checks the active sets of all of $p$'s locks, and compares $p$'s priority to all \descriptor{s} $q$ in those sets such that $q$'s status is \code{active}. On each such comparison, the \descriptor{} with the lower priority is \emph{eliminated}. This means that its status is atomically CASed from \code{active} to \code{lost}. After comparing $p$'s priority with all \descriptor{s} in the active sets of all of its locks, the \code{run} function \emph{decides} whether $p$ won or lost. This involves trying to atomically CAS $p$'s status to \code{won}. This will work if and only if $p$ hasn't been previously eliminated. Finally, \code{run(p)} `celebrates' the end of $p$'s competitions by running its thunk if its status is \code{won}. The celebrationIfWon is also executed for each competitor that $p$ faced. This ensures that any \descriptor{} that reaches the \code{won} status gets its thunk executed before another \descriptor{} sharing a lock wins, and ensures mutual exclusion.

\begin{figure}
\begin{lstlisting}
struct @\Descriptor@:
	ActiveSet* lockList //list of active set objects
	thunk
	int priority
	status = {active, won, lost}
	
	bool getFlag():
		return (priority > 0)
	void setFlag():
		Delay until @$T_0 = c\cdot \kappa^2 \cdot L^2\cdot T$@ total steps taken @\label{line:delay1}@ 
		priority = rand //@\textit{reveal step of p}@
	void clearFlag():
		priority = -1

tryLocks(lockList, thunk):  
	p = new @\Descriptor@(lockList, thunk, -1, active)
	for each lock @$\ell$@ in p.lockList:
		set = getSet(@$\ell$@) @\label{line:getHelpSet}@
		for each p' in set:
				run(p')
	multiInsert(p, p.lockList) @\label{line:insert}@
	run(p)
	multiRemove(p, p.lockList)
	Delay until @$T_1= c'\cdot \kappa \cdot L\cdot T$@ steps taken since previous delay @\label{line:delay2}@
	
run(@\Descriptor@ p): 
	for each lock @$\ell$@ in p.lockList:
		set = getSet(@$\ell$@) @\label{line:runGetSet}@
		if (p.status == active): @\label{line:checkActive}@
			for p' in set:
				if (p'.status == active): 
					if  p.priority > p'.priority: @\label{line:compare}@
						eliminate(p')           @\label{line:eliminate1}@     
					else if (p != p'): eliminate(p) @\label{line:eliminate2}@
				celebrateIfWon(p') @\label{line:celebrateOther}@
	decide(p)  @\label{line:decide}@ 
	celebrateIfWon(p) @\label{line:runCelebrateSelf}@
	
decide(@\Descriptor@ p):
	CAS(p.status, active, won)

eliminate(@\Descriptor@ p):
	CAS(p.status, active, lost)
	
celebrateIfWon(@\Descriptor@ p):  
	if(p.status == won):
		execute p.thunk 
\end{lstlisting}
\caption{Lock Algorithm}
\label{alg:lock}
\end{figure}


	\begin{theorem}\label{thm:tryLocks}
		Let $\kappa$ be the maximum point contention any single lock can experience.  Let $L$ be the maximum number of locks in the lock set of any descriptor. Let $T$ be the maximum length of a thunk.
		Algorithm~\ref{alg:lock} provides 
		fine-grained locks 
		such that the number of steps per tryLock attempt is $O(\kappa^2 L^2 T)$.
	\end{theorem}

\begin{proof}
	It is easy to see this theorem holds by observing the tryLock and the multi active set algorithms. Each call to getSet takes an number of steps linear in $\kappa$, and each \minsert{} and \mremove{} takes $O(\kappa L)$ steps. Each instance of the \code{run} method calls getSet $O(L)$ times, and executes $O(T)$ steps for each \descriptor{} in the resulting sets. Since the \code{run} method is called $O(\kappa L)$ times in a tryLock, this leads to the total step complexity of $O(\kappa^2 L^2 T)$.
\end{proof}

\paragraph{Delays.} The algorithm as described thus far captures the essence of the approach; clear out any competitors whose priorities could have had an effect on your start time, and then compete by inserting yourself into the active sets of your locks and comparing your priority to all others. However, it also has weak points that the adversary can exploit to skew the priority distribution of the competitors of certain \descriptor{s}. In particular, a \descriptor{} $p$ takes a variable amount of its own steps to get to its reveal point, and a variable amount of steps after that to finish its attempt. This variance is caused by the amount of contention it experiences -- how many \descriptor{s} are accessing the active set or are in it when $p$ accesses the same active set, and what are their priorities. The number of other \descriptor{s} affect the time its insertion into the active sets takes (as shown in Section~\ref{sec:activeset}), as well as the number of \descriptor{s} it must compare its priority to. Furthermore, if $p$ runs a \descriptor{'s} thunk, this could take longer than if it simply eliminated it. The adversary can use this variance to skew the distribution of priorities of \descriptor{s} that $p$ competes against. 


To avoid this, we inject \emph{delays} at two critical points in the algorithm. The first is immediately before $p$'s reveal step. The goal is to ensure that $p$ always takes a fixed number of steps from its start time until its reveal step. This means that once the adversary chooses to start $p$, it has also chosen its reveal time, and cannot modify this after discovering more information.\footnote{For example, after $p$'s start time, it's possible that some descriptor that started before $p$ reaches its reveal time. At this point, the adversary has more information about $p$'s competitors. It can attempt to extend $p$'s time before its reveal step by starting new descriptors and forcing $p$ to help them. We want to avoid this possibility.} To achieve this goal, we choose a fixed number of steps until $p$'s reveal step that is an upper bound on the amount of time $p$ can take to arrive at its reveal step; $T_0 = c \cdot \kappa^2 L^2 \cdot T$, where $\kappa$ is the maximum point contention on any lock, $L$ is the maximum number of locks per tryLock attempt, $T$ is the maximum number of steps to run a single thunk, and $c$ is any any sufficiently large constant. 

Similarly, we introduce a delay after $p$'s run to ensure that the time between its reveal step and termination is also determined at its invocation. Here, there is no need to square $\kappa$ and $L$, since $p$ only needs to execute \code{run} for itself after its reveal step. Therefore, $T_1 = c \cdot \kappa L T$ for some sufficiently large constant $c$. 

 It is important to note that delay is in terms of steps for a particular process.
The scheduler can run different processes at very different rates, so the delay counted in total number of steps across all processes in the history on one process could be very different than on another depending on the scheduler.
  
\subsection{Safety and Fairness}


We show that Algorithm~\ref{alg:lock} is correct by showing that it satisfies the mutual exclusion with idempotence property (Definition~\ref{def:mutex}). The key to its correctness is in the way that the \code{run} function works. In particular, a \descriptor{}'s status can change at most once. Furthermore, \code{celebrateIfWon} never actually runs a thunk unless it status is \code{won} (at which point it cannot lose anymore). Since its status can become \code{won} only in Line~\ref{line:decide}, after it compares its priority to that of the \descriptor{}s on all of its locks, and also celebrates any winners out of these \descriptor{s}, by the time it celebrates for itself, the thunks of any earlier winners on any of its locks have already been executed. Thus, the placements of the celebrations (once on Line~\ref{line:celebrateOther} for its competitors, and once on Line~\ref{line:runCelebrateSelf} for itself) are crucial for the safety of the algorithm. The full safety proof for the algorithm is presented in the full version of the paper.

We now focus on proving the fairness guarantees of the algorithm. In essence, we show that the adversary's power is quite limited. In particular, the adversary must decide whether or not two \descriptor{s} could threaten each other (i.e. their priorities could be compared in Line~\ref{line:compare}) before learning 
any information on either of their priorities. Intuitively, this is due to two main reasons.

The first is because of the helping mechanism. Before a \descriptor{} $p$ reveals its priority, it puts all \descriptor{s} whose priority was already revealed in a state in which they can no longer threaten it -- their status becomes non-active. 

To show this property formally, we begin with establishing some terminology; we say a \descriptor{} $p$ \emph{causes a \descriptor{} $p'$ to fail} if $p'$'s status is changed to \code{lost} in the \code{eliminate(p')} call on Line~\ref{line:eliminate1} or~\ref{line:eliminate2} after comparing $p'$'s priority with $p$'s priority. We say a \descriptor{} $p$ \emph{can cause $p'$ to fail} if $p$ and $p'$'s priorities are compared on Line~\ref{line:eliminate1} or~\ref{line:eliminate2} during the execution. We can now discuss when \descriptor{s} can and cannot cause each other to fail, in the next two useful lemmas.

\begin{lemma}\label{lem:differentLocks}
	A \descriptor{} $p$ cannot cause a \descriptor{} $p'$ to fail if their lock sets do not intersect.
\end{lemma}

\begin{proof}
	A \descriptor{} $p$ is only inserted into the sets corresponding to locks in its lock set (Line~\ref{line:insert}). Furthermore, in \code{run(p)}, only the sets of locks in $p$'s lock sets are compared against. Therefore, a \descriptor{} $p'$ will never be compared against and potentially eliminated by a \descriptor{} $p$ if their lock sets do not intersect.
\end{proof}

\begin{lemma}\label{lem:endThreat}
	A \descriptor{} $p$ cannot cause a \descriptor{} $p'$ to fail if $p$'s status stopped being \code{active} before $p'$'s reveal step. 
\end{lemma}

\begin{proof}
	First note that before $p'$'s reveal step, no \descriptor{} will eliminate $p'$, since its priority will be negative, and the comparison with it will be skipped on Line~\ref{line:compare}. Therefore, $p$ cannot cause $p'$ to fail before $p'$'s reveal step. By Lemma~\ref{lem:stableStatus}, $p$'s status will never be \code{active} again after it stops being \code{active}. Therefore, in any \code{run(p)} or \code{run(p')} call, the comparison of $p$ with any other \descriptor{} will be skipped on Line~\ref{line:checkActive} after $p$ becomes inactive. In particular, this comparison will always be skipped after $p'$'s reveal step.
\end{proof}

Equipped with the above two lemmas, we can now show that the helping mechanism has the effect we want; \descriptor{s} whose intervals only overlap before one of their reveal steps cannot cause each other to fail.

\begin{lemma}\label{lem:startThreat}
	Let $p$ and $p'$ be  \descriptor{s} such that $p$'s tryLock starts after $p'$'s reveal step. Neither \descriptor{} can cause the other to fail.
\end{lemma}

\begin{proof}
	By Lemma~\ref{lem:differentLocks}, if $p$ and $p'$'s lock sets do not overlap, the lemma holds.
	Otherwise, if $p'$ has already removed itself from the active set by the time $p$ started its \code{getSet} on Line~\ref{line:getHelpSet}, then $p'$ must have already won or lost by this time. In particular, $p$ wasn't in the active set during $p'$'s \code{run(p')}, and therefore could not have caused $p'$ to fail. Furthermore, by Lemma~\ref{lem:endThreat}, $p'$ cannot cause $p$ to fail.
	
	So, assume that $p'$ is still in the active set at the time $p$ started its \code{getSet} on Line~\ref{line:getHelpSet}. 
	By the set regularity of the multi active set algorithm,
	since $p'$ must have completed its insertion into the active set before $p$ started its \code{getSet} on Line~\ref{line:getHelpSet}, $p$ must have $p'$ in the set it gets. 
	Therefore, $p$ calls \code{run(p')} before its own reveal step, so by Lemma~\ref{lem:doneStatus}, $p'$ wins or fails because of a different \descriptor{}. Furthermore, again by Lemma~\ref{lem:endThreat}, $p'$ cannot cause $p$ to fail, since a \code{run(p')} call completes before $p$'s reveal step, and therefore $p'$ is no longer active by that time.
\end{proof}

That is, in this lemma we show that if $p$ starts after $p'$'s reveal step, their priorities can never be compared. 
We can now introduce some more terminology to help us reason further about fairness. We define the \emph{interval} of a \descriptor{} $p$ as the time between its calling process's call to tryLock and the time at which its tryLock call returns. Furthermore, a \descriptor{}'s \emph{threat interval} is the time between the beginning of its interval and the time at which its status stops being \code{active}.
We define $p$'s \emph{threateners} as the set of \descriptor{s} that can cause $p$ to fail. 
We make the following observations about the relationships between \descriptor{s}.

\begin{observation}\label{obs:threat}
	The set of \descriptor{s} whose intervals overlap a \descriptor{} $p$'s reveal step includes all of $p$'s threateners.
\end{observation}

\begin{proof}
	Every \descriptor{} $p$'s interval includes a complete call to \code{run(p)}. Thus, by Lemma~\ref{lem:doneStatus}, the status of any \descriptor{} is not \code{active} by the end of its interval. The lemma is therefore immediately implied from Lemmas~\ref{lem:endThreat} and~\ref{lem:startThreat}.
\end{proof}

\begin{observation}\label{obs:startPriority}
	At the time at which a \descriptor{}'s interval starts, none of its threateners have reached their reveal step.
\end{observation}

%

The other property of the algorithm which ensures that the adversary cannot pit \descriptor{s} against one another after knowing their priorities stems from the delays in the algorithm. In particular,

\begin{observation}\label{obs:constantSteps}
	Each \descriptor{} interval takes the same number of steps by the initiating process between its start and its reveal step, and between its reveal step and the end of its interval, regardless of the schedule or randomness.
\end{observation}


Together, the helping mechanism and the delays allow us to prove the main lemma for the fairness the argument. This lemma relies on the notion of \emph{potential threateners}. 
We say  that a \descriptor{} $p$ is a \emph{potential threatener} of another \descriptor{} $p'$ if (1) $p$'s interval overlaps with $p'$'s reveal step, and (2) $p'$ did not execute a \code{run(p)}. Note that by Observation~\ref{obs:threat} and Lemma~\ref{lem:endThreat}, the set of potential threateners of a \descriptor{} is a superset of its actual threateners. 


\begin{lemma}\label{lem:pairwiseIndep}
		The player adversary has no information on the priorities of $p$ or $p'$ at the time at which it makes $p'$ threaten $p$.
\end{lemma}

\begin{proof}
	By Observation~\ref{obs:constantSteps}, once a \descriptor{} starts, its reveal step and last step of its interval are determined. Since these two steps are what determines whether a \descriptor{} will be a potential threatener of  another \descriptor{}, the start times of the two \descriptor{s} determine whether this event occurs. 
	Furthermore, by Lemma~\ref{lem:startThreat} and the definition of potential threateners, if $p$ is a potential threatener of $p'$ in an execution $E$, then both $p$ and $p'$ must have started their tryLock interval before either of their reveal steps. Therefore, the player adversary had no information on their priorities at the time at which it decided to start their intervals. 
\end{proof}

The final theorem is easily implied from this lemma by recalling that the choice of priorities of each \descriptor{} is always done uniformly at random and independently of the history so far. Thus, the adversary can choose whether or not to introduce more threateners for a \descriptor{} $p$, but cannot affect their priorities. Since there is a bound on the amount of contention the adversary can introduce, we get a bound on $p$'s chance of success.

\begin{theorem}\label{thm:probSuccess}
	Let $k_\ell$ be the bound on the maximum point contention possible on lock $\ell$, and let $C_p = \sum_{\ell \in p.lockList} k_\ell$ be the sum of the bounds on the point contention across all locks in a \descriptor{} $p$'s lock list. 
	Algorithm~\ref{alg:lock} provides wait-free fine-grained locks against an oblivious scheduler and an adaptive player such that the probability that $p$ succeeds in its tryLock in $A$ is at least $\frac{1}{C_p}$.
\end{theorem}

\begin{proof}
	On each lock $\ell$ in $p$'s lock list, the adversary can make at most $k_\ell$ \descriptor{s} be potential threateners of $p$. Assume that all priorities of the \descriptor{s} are picked uniformly at random, but the priority of a given \descriptor{} $p'$ is hidden until after the adversary chooses whether or not $p'$ will be a potential threatener of $p$. This is equivalent to our setting since the priorities are always picked uniformly at random, and, by Lemma~\ref{lem:pairwiseIndep}, the adversary has no information on a \descriptor{} $p'$'s priority until after it decides whether it will potentially threaten $p$. Once the adversary discovers the priority of a \descriptor{}, it can decide whether the next \descriptor{} will be a potential threatener of $p$ and then reveal the corresponding priority. In the worst case, the adversary can reveal up to $C_p$ of those predetermined priorities. Since the set of potential threateners of $p$ include all actual threateners of $p$, this makes $p$ threatened by $C_p$ uniformly chosen random values in the worst case, giving it a $\frac{1}{C_p}$ probability of having the maximum priority of all of them.
\end{proof}

Note that the theorem is stated using the point contention bounds of the specific locks that are in the lock set of tryLock attempt $p$. In terms of the general bounds $\kappa$ on the point contention per lock and $L$ on the number of locks in any lock set, the probability of success can be bounded from below at $\frac{1}{\kappa L}$.  Theorem~\ref{thm:probSuccess} and Theorem~\ref{thm:tryLocks} together imply Theorem~\ref{thm:main}.

\paragraph{Using the multi active set implementation.}
We note that as shown by Golab et al.~\cite{GHW11}, using implemented rather than atomic objects in a randomized algorithm can affect the probability distributions that an adversarial scheduler can produce. This effect can occur when several operations on the implemented objects are executed concurrently. Thus, we must be careful when using our set regular multi active set object in our randomized lock implementation. However, the way in which our lock algorithm uses the multi active set, and the way we use it in our analysis, is not subject to this effect. To see this, 
note that there is slack in our analysis; we consider the set of \descriptor{s} with \emph{potential to compete}, where this means that a \code{getSet} executed by one \descriptor{} \emph{could} see the other \descriptor{}. That is, any change in priority distributions that the adversary could try to achieve is already covered by our analysis.

\subsection{Handling Unknown Bounds}

So far, we've been assuming that $\kappa$, the upper bound on the point contention of any lock, and $L$, the upper bound on the number of locks per tryLock, are known to the lock algorithm. In this subsection, we briefly outline how to handle these bounds being known to the adversary but not the algorithm. The full algorithm and analysis for this case appear in the full version of this paper.

Algorithm~\ref{alg:lock} used these bounds in two ways: firstly, the active set objects were instantiated with arrays of size $\kappa$, and secondly, $\kappa$ and $L$ were used to determine how many delay steps each tryLock attempt must take to ensure that each \descriptor{}'s reveal step and final step of the attempt always happen after the same number of steps since the attempt's start time. The first concern is easy to fix by using more space; instead of setting the size of the announcement array of the active set object to $\kappa$, we set it to $P$, the total number of processes in the system. In most applications, this number is significantly larger than $\kappa$.  We note that the size of each individual set pointed at by the array slots is still at most~$\kappa$, and therefore the time bounds remain proportional to $\kappa$ as well.

However, the second problem is more challenging, as the delays are crucial for the fairness bounds to hold. We make a few key changes to the algorithm. Firstly, must ensure that the size of the active set read on line~\ref{line:runGetSet} by an attempt $p$ is fixed before the adversary discovers the priority of $p$. To do so, we split the reveal step into two parts; the \emph{participation-reveal step}, and the \emph{priority-reveal step}. The participation-reveal step occurs after a \descriptor{} $p$ inserts itself into the active set object of each of its locks. In this step, it changes its priority from $-1$ to a special \code{TBD} value, indicating that it is ready to participate in the lock competition, but not yet revealing its priority. At this point, all locks are queried to obtain their active sets, and only then is the priority of $p$ revealed. The key insight is that after the priority is revealed, the active set objects are no longer queried, and instead the local copies of the sets, obtained just before the priority reveal step, are used. Thus, the adversary does not learn $p$'s priority until after it can no longer affect the set of $p$'s potential threateners. 

However, there is still the issue of the steps a \descriptor{} executes before its participation-reveal step. In the first part of its execution, a \descriptor{} must help others to complete their \code{run} call, and must then call \minsert{} on itself and its lock set. The length of these tasks vary depending on the number of active \descriptor{s} in the system, and can therefore be controlled by the player adversary.
Instead of relying on $\kappa$ and $L$, we employ a \emph{doubling} trick; $p$ measures the number of steps it took until right before its participation-reveal step, and then employs a delay to bring that number up to the nearest power of two. In this way, while the adversary still has control of the number of steps $p$ will take, this number is now guaranteed to be one of only 
$\log (\kappa L T)$ values.
We arrive at the following result.

\begin{theorem}
		Let $k_\ell$ be the bound on the maximum point contention possible on lock $\ell$, $\kappa$ be an upper bound on $k_\ell$ for all $\ell$, and let $C_p = \sum_{\ell \in p.lockList} k_\ell$ be the sum of the bounds on the point contention across all locks in a \descriptor{} $p$'s lock list. Furthermore, let $L$ be the maximum number of locks per tryLock attempt, and $T$ be the maximum length of a critical section. Then there exists an algorithm $A$ for wait-free fine-grained locks against an oblivious scheduler and an adaptive player such that (1) the probability that $p$ succeeds in its tryLock is at least $\frac{1}{C_p\log (\kappa L T)}$ in $A$, and (2) $A$ that does not know the bounds $k_\ell$, $\kappa$ and $L$. 
\end{theorem}

	\section{Discussion}

In this paper, we present an algorithm for randomized wait-free locks that guarantees each lock attempt terminates within $O(\kappa^2 L^2 T)$ steps and succeeds with probability $1/\kappa L$ where $\kappa$ is an upper bound on the contention on each lock, $L$ is an upper bound on the number of locks acquired in each lock attempt, and $T$ is an upper bound on the number of steps in a critical section. 
We further show a version of the algorithm that does not require knowledge of $\kappa$ and $L$, where the success probability is reduced by a factor of $O(\log(\kappa L T))$.

There are several interesting directions for further study. In particular, 
it is possible that a lock algorithm exists that reduces the number of steps per attempt to  $O(\kappa L T)$. Furthermore, while the success probability of each attempt in our algorithm adapts to the true contention, our step bounds instead depend on the given upper bounds. It would be interesting to develop a lock algorithm that adapts its step complexity to the actual contention. To achieve such adaptiveness, an algorithm would necessarily have to avoid the delays that we use;  our bounds rely on injecting fixed delays in which processes must stall if they finish an attempt `too early'. 

It would be interesting to see how well our proposed lock algorithm does in practice. It has recently been shown that lock-free locks can be practical~\cite{ben2022lock}, and we believe that the stronger bounds that our construction provides may be useful. 
In real systems, allowing the nesting of locks may be a useful primitive. While our construction allows acquiring multiple locks, these locks must be specified in advance and cannot be acquired from within a thunk (critical section).  We believe that our algorithm would maintain safety if locks were nested, but its proven bounds would not hold. It would therefore be interesting to develop an algorithm for wait-free nested locks with strong bounds.

\section*{Acknowledgements}
	We thank the anonymous referees for their comments. This work was supported by the National Science Foundation grants CCF-1901381, CCF-1910030, and CCF-1919223.

	\bibliographystyle{plain}
	\balance
	\bibliography{strings,biblio}
	
\end{document}